\pdfoutput=1
\documentclass[11pt]{article}

\usepackage[margin=1in]{geometry}
\usepackage{microtype}
\usepackage{graphicx}
\usepackage{amsmath,amssymb,amsfonts,amsthm,mathrsfs,bm}
\usepackage{algorithm}
\usepackage{algpseudocode}
\usepackage[authoryear,round]{natbib}
\usepackage[hidelinks]{hyperref}
\usepackage{xr-hyper}

\IfFileExists{SS_supplement.aux}{%
  \externaldocument[][nocite]{SS_supplement}%
}{}

\theoremstyle{plain}
\newtheorem{theorem}{Theorem}

\newtheorem{proposition}[theorem]{Proposition}

\theoremstyle{remark}

\theoremstyle{definition}

\numberwithin{equation}{section}
\DeclareMathOperator*{\argmin}{\arg\,\min}
\newcommand{\dd}{\mathrm{d}}
\newcommand{\ys}{\mathbf{y}^\star}
\newcommand{\xbf}{\mathbf{x}}
\newcommand{\ybf}{\mathbf{y}}

\newcommand{\Ybf}{\mathbf{Y}}

\newcommand{\ubf}{\mathbf{u}}
\newcommand{\wbf}{\mathbf{w}}
\newcommand{\bbf}{\mathbf{b}}
\newcommand{\Abf}{\mathbf{A}}
\newcommand{\Ibf}{\mathbf{I}}

\newcommand{\mbf}{\mathbf{m}}
\newcommand{\Sbf}{\mathbf{S}}
\newcommand{\sbf}{\mathbf{s}}

\algrenewcommand\algorithmicrequire{\textbf{Input:}}
\algrenewcommand\algorithmicensure{\textbf{Output:}}

\title{Coresets for Density-Based Clustering\\
via Bayesian Predictive Partition Matching}
\author{Bernardo Flores\\[0.4em]
\small Department of Biostatistics \& Informatics\\
\small University of Colorado Anschutz Medical Campus\\[0.2em]
\small \texttt{bernardo.floreslopez@cuanschutz.edu}}
\date{}

\begin{document}
\maketitle

\begin{abstract}
Density-based clustering identifies groups as connected high-density regions, accommodating irregular shapes and reflecting their geometry, but this reliance on density estimation introduces uncertainty into the resulting partitions, which is often costly to quantify. To represent this uncertainty efficiently, we construct a compact representation of the data, or coreset, designed to preserve the distribution of predictive density-based partitions. We sample observations from a Dirichlet process surrogate, smooth them to estimate densities, and use the resulting level sets to partition a common reference set. Starting from a uniform subsample, we learn an affine transformation that improves agreement with full-data partitions by minimizing a tractable approximation to paired expected variation of information. Fresh, independent predictive simulations then guide the selection of the final transformation. Our analysis provides approximation and selection guarantees and establishes consistency for ideal fixed-resolution partitions under geometric conditions. Simulations and an application to mass cytometry show improved partition agreement over uniform subsampling, often matching substantially larger supports.
\end{abstract}

\noindent\textbf{Keywords:}
Bayesian nonparametrics; coresets; density-based clustering; partition laws; variation of information; Wasserstein distance.

\section{Introduction}
\label{sec:intro}
Cluster analysis seeks to identify groups of observations that are
homogeneous within groups and distinct across them.
The interpretation of a cluster, and consequently the uncertainty
associated with it, depends on how clusters are defined.
Two broad probabilistic approaches are particularly relevant: model-based
clustering, which defines groups through latent allocations or
partitions, and density-based clustering, which identifies groups
through the geometry of the underlying density.

Model-based clustering represents groups through latent
allocations or partitions, including those induced by mixture
models and more general partition constructions \citep{wade2023}.
In mixture models, observations belong to the same cluster when they
share a component, even though mixture components do not necessarily
correspond to distinct groups in the data. A single group may require
several components, and under Bayesian inference, the posterior on
the component count need not concentrate on the number of groups
\citep{miller2014}.

Density-based clustering instead defines clusters as the connected
components of density upper level sets \citep{hartigan1975,Chacon2015}.
This geometric definition accommodates irregular shapes without
equating clusters with mixture components, and its statistical
properties are well studied \citep{hartigan1981,rinaldo_stability_2010}.
However, small errors in density estimation can alter level-set
connectivity, introducing uncertainty about cluster memberships. This uncertainty can be assessed through approaches including
the bootstrap \citep{kim2016}, Bayesian inference on the density
\citep{buch2026}, and martingale posterior constructions
\citep{fong2024,bariletto2026}.
When applied to density-based clustering, these approaches
may require repeated density estimation and partition construction,
which can be costly for large datasets.
Subsampling reduces this burden but may overlook
small groups, as illustrated in
Figure~\ref{fig:subsampling_failure}.
This motivates compact data representations that
preserve the distribution of resulting partitions.

\begin{figure}[t]
\centering
\includegraphics[width=0.85\textwidth]{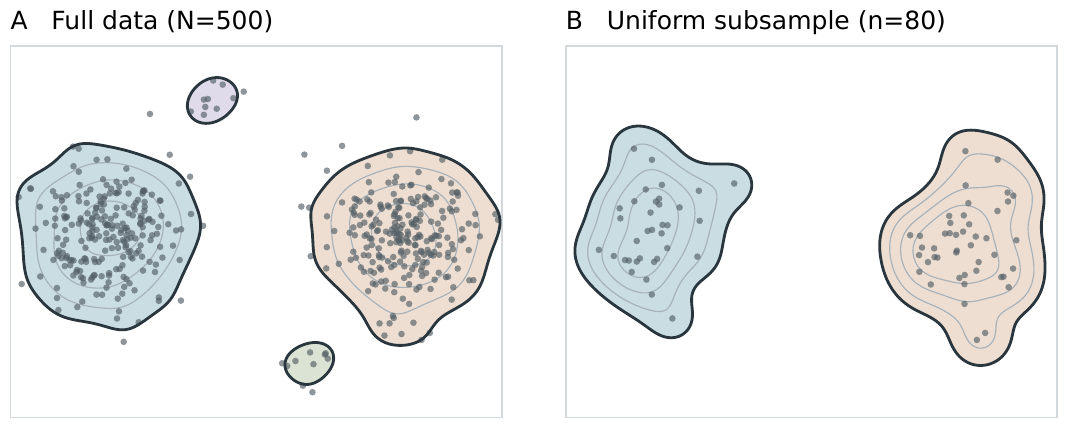}
\caption{Effect of uniform subsampling on density estimation. Panel (a) shows the full dataset, while
panel (b) shows a uniform subsample that misses two
rare groups.}
\label{fig:subsampling_failure}
\end{figure}

A coreset is a compact representation of a dataset designed to
preserve quantities relevant to an analysis while reducing
computational costs \citep{huggins2016,campbell_automated_2019}.
Existing coreset constructions target different quantities: Bayesian
coresets approximate likelihoods or parameter posteriors,
while clustering coresets preserve objectives such as
mixture likelihoods \citep{lucic2018}, Dirichlet process means
\citep{bachem2015}, or kernel density estimates
\citep{phillips2020}. Wasserstein coresets approximate
empirical distributions \citep{claici_wasserstein_2020}.
Other scalable density-based clustering methods accelerate
DBSCAN \citep{ester96} using subsets of observations
\citep{jang2019}. These approaches may indirectly preserve clustering,
but do not target it directly.

We propose a coreset construction that aims to preserve the
distribution of density-based partitions, rather than
partitions defined by latent mixture allocations.
To characterize this distribution without fitting a
complex density model, we follow
\citet{lyddon_nonparametric_2018} and use a Dirichlet
process (DP) to approximate the unknown sampling
distribution. Its data-calibrated base measure guides
prediction, while the observed data inform the
conditional predictive distribution.
We generate datasets from this distribution using
the DP's P\'olya urn scheme \citep{blackwell1973},
smooth them to estimate densities, and identify clusters
through their upper level sets.
Each density induces a partition of a common reference
set, allowing us to compare the full-data and coreset
distributions of partitions. The reference observations
guide fitting and selection but are excluded from
predictive conditioning.

Our coreset consists of a uniform subsample whose
locations are adjusted through an affine transformation.
We learn this transformation to approximate the
full-data distribution of density-based partitions.
We measure discrepancies between the full-data and
coreset partition distributions using the Wasserstein
distance, with variation of information (VoI)
\citep{melia,wade2018} as the ground metric.
Since directly optimizing this distance is
computationally difficult, we develop a tractable
fitting procedure and select the final transformation
using independent Monte Carlo replicates.

Section~\ref{sec:coresets} develops the proposed construction,
and Section~\ref{sec:theory} presents its theoretical guarantees.
Section~\ref{sec:experiments} evaluates partition agreement
and construction cost using simulations and mass cytometry
data. Further details and additional experiments appear in
the Supplementary Material. Section~\ref{sec:discussion}
discusses the method's implications and limitations.

\section{Background and methodology}
\label{sec:coresets}

\subsection{Bayesian coresets}
\label{sec:corebg}

Conventional Bayesian coresets approximate the full-data
posterior using a sparse, nonnegative weighted log-likelihood
\citep{huggins2016,campbell_automated_2019}.
For observations $y_{1:N}$ and parameter $\theta$, the
weighted posterior is
\[
\pi_{\wbf}(\theta\mid y_{1:N})
\propto \pi_0(\theta)
\exp\left\{\sum_{i=1}^N w_i\log p(y_i\mid\theta)\right\},
\]
where $w_i\geq 0$ and only $n\ll N$ weights are nonzero.
The nonzero weights select a small subset of observations $\ys_{1:n}$, while
their magnitudes rescale the corresponding log-likelihood
contributions. These weights are optimized to minimize a discrepancy
between the weighted and full-data posteriors.
Different versions of this construction either select
and weight observations or fix the support and optimize
its weights \citep{naik2022fast,chen2022shf}.
Pseudocoresets instead optimize synthetic support locations
\citep{manousakas2020}. We use unweighted uniform
subsampling as a basic benchmark.

Wasserstein coresets instead minimize a discrepancy
between full-data and reduced empirical distributions,
without requiring a likelihood model
\citep{claici_wasserstein_2020}.
For probability measures $\mu$ and $\nu$ on a metric space
$(\mathcal X,\mathtt d)$, the $p$-Wasserstein distance is \citep{villani2003topics}
\[
W_p(\mu,\nu)^p
=\inf_{\gamma\in\Pi(\mu,\nu)}
\int \mathtt d(\xbf,\ybf)^p\,\gamma(\dd \xbf,\dd \ybf),
\]
where $\Pi(\mu,\nu)$ denotes the set of couplings between
$\mu$ and $\nu$. In this work, we construct a coreset by
applying an affine transformation to a uniform subsample
$\ys_{1:n}$. We also use a Wasserstein discrepancy,
but between distributions of density-based partitions rather than empirical
distributions of observations. Section~\ref{sec:coresets} develops our construction.

\subsection{Nonparametric predictive construction}
\label{sec:proxy_npl}

To preserve uncertainty about density-based partitions,
we need to account for variation in the underlying density
rather than rely on a single estimate. Instead of fitting
a complex probabilistic density model, we follow the
nonparametric learning approach of
\citet{lyddon_nonparametric_2018}, using a Dirichlet
process (DP) \citep{ferguson_bayesian_1973} to approximate
the unknown sampling distribution.
A DP defines a distribution over almost surely discrete
probability measures, with a base measure determining their
mean and a concentration parameter controlling their variability.
We smooth draws from the DP predictive distribution
to estimate densities that partition a common reference
set (Section~\ref{sec:proxy}).

Let $\Ybf\in\mathbb R^{N\times d}$ denote a data matrix
with rows $\ybf_1,\ldots,\ybf_N$, comprising $d$
continuous variables suitable for kernel density estimation. Throughout, bold uppercase letters denote matrices,
and bold lowercase letters denote vectors.
We reserve $R$ observations as a reference set and use
the remaining $N-R$ for analysis;
Section~\ref{sec:settings} discusses the choice of $R$.
The subscripts $N$ and $n$ distinguish the full-data
and coreset constructions, although the former conditions
on $N-R$ observations.
The coreset consists of a uniform subsample $\ys_{1:n}$,
with $n\ll N$, and an affine transformation
$T\in\mathcal T$, where
$T(\ybf)=\Abf\ybf+\bbf$,
$\Abf\in\mathbb R^{d\times d}$, and
$\bbf\in\mathbb R^d$.

The surrogate places a DP directly on probability
measures in observation space without restricting the sampling distribution
to a parametric family. Conditional on $\phi$, we write
\[
\Ybf_i\mid P_\phi,\phi\overset{\mathrm{iid}}{\sim}P_\phi,
\qquad
P_\phi\mid\phi\sim\mathrm{DP}(\alpha,F_\phi),
\]
where $F_\phi$ is a working base measure and $\alpha>0$
is a fixed concentration parameter.

For a conditioning sample $\ybf_{1:m}$, DP conjugacy gives
\begin{equation}\label{eq:postbase}
P_\phi\mid\ybf_{1:m},\phi
\sim\mathrm{DP}(\alpha+m,F_\phi^m),
\qquad
F_\phi^m
=
\frac{\alpha}{\alpha+m}F_\phi
+
\frac{1}{\alpha+m}\sum_{i=1}^m\delta_{\ybf_i}.
\end{equation}
Conditional on $\phi$, $F_\phi^m$ is both the posterior
mean and the one-step predictive distribution.
The empirical distribution and $F_\phi$ contribute
weights $m/(\alpha+m)$ and $\alpha/(\alpha+m)$,
respectively. Thus, $F_\phi$ should provide a reasonable
predictive envelope, especially for small $m$, but need
not capture every mode because the empirical component
retains the observed structure. Unlike the Bayesian bootstrap \citep{rubin_bayesian_1981},
which assigns mass only to observed support, the DP retains
predictive mass beyond that support through $F_\phi$.
As $\alpha\downarrow0$, the conditional posterior
approaches the Bayesian bootstrap; as
$\alpha\to\infty$, it concentrates around $F_\phi$.

We specify the working base measure as
$F_\phi=\mathcal N_d(\bm\mu_\phi,\bm\Sigma_\phi)$,
with $\bm\mu_\phi\sim\mathcal N_d(\mbf_0,3\Ibf_d)$ and
$\bm\Sigma_\phi\sim\mathcal{IW}(d+2,\Sbf_0)$,
where $\mbf_0$ and $\Sbf_0$ are the sample mean
and covariance of the analysis set, respectively.

This calibration aims to make the base measure more
representative of the sampling distribution in location
and scale, without assuming Gaussian sampling
\citep{fraley2007}. Importantly, the empirical calibration
does not affect our consistency guarantees, since the relevant bounds
hold uniformly over the base measure (Supplementary
Section~S1.6, Proposition~\ref{prop:camerlenghi_hyperprior}
and Remark~\ref{rem:data_hyper}).
The Gaussian base measure is convenient but not required.
Section~\ref{app:implementation} of the Supplementary
Material discusses the implemented specification.
The full-data and coreset constructions use the same
concentration parameter, whose choice is discussed in
Section~\ref{sec:settings}.

To allow variation in the working centering measure,
we average the conditional Dirichlet process posteriors
over a specified distribution $\pi_\phi$, obtaining
the surrogate distribution over random sampling measures
$\widetilde\Xi_m=\int\mathrm{DP}(\alpha+m,F_\phi^m)\,\pi_\phi(\dd\phi)$.
Here, $\pi_\phi$ denotes the mixing distribution
after any empirical calibration; it is not a
likelihood-updated posterior for $\phi$. The observations
update the conditional DP through $F_\phi^m$; varying
$\phi$ supplies different centering measures. Each paired
Monte Carlo replication shares a draw $\phi\sim\pi_\phi$
across the full-data and coreset predictive trajectories.

We generate predictive trajectories from the DP using
the P\'olya urn scheme \citep{blackwell1973}.
Given the conditioning sample and the first $j-1$
predictive observations, the next draw follows
\begin{equation}\label{eq:urn}
\widetilde \Ybf_j
\mid \ybf_{1:m},\widetilde \Ybf_{1:j-1},\phi
\sim
\frac{
\alpha F_\phi
+\sum_{i=1}^m\delta_{\ybf_i}
+\sum_{j'=1}^{j-1}\delta_{\widetilde \Ybf_{j'}}
}{
\alpha+m+j-1
},
\qquad j=1,\ldots,L.
\end{equation}
The final term reinforces previous predictive draws,
so observations within a trajectory are dependent
after integrating out the random sampling measure.
Equivalently, each trajectory can be viewed as
conditionally i.i.d.\ sampling from a realization
of the DP posterior. As $L$ increases, its empirical
measure converges weakly to that realization,
preserving variation across posterior draws.
Section~\ref{sec:settings} discusses the choice of $L$.

Kernel smoothing of predictive trajectories provides
the continuous densities needed for clustering.
For a fixed bandwidth, as the trajectory length increases,
the resulting KDE converges to the convolution of the
underlying random measure with the kernel.
Under the stated conditions, the Supplementary Material
establishes that, as the conditioning sample size grows,
the smoothed DP surrogate concentrates uniformly on the
compact domain around the smoothed sampling distribution
$K_h*P_0$ (Proposition~\ref{prop:camerlenghi_hyperprior}).
We also control the error from finite predictive trajectories
(Lemma~\ref{lem:urn_kde_concentration}).

\subsection{Predictive partitions and reference-set comparison}
\label{sec:proxy}

Fitting the coreset requires comparing the density-based
partitions induced by full-data and coreset predictive
trajectories. Because these trajectories contain different
observations, we evaluate their partitions on a common
reference set.

For a full-data predictive trajectory
$\{\widetilde\ybf_j\}_{j=1}^L$, the kernel
density estimator (KDE) is
\begin{equation}\label{eq:kde}
\widehat f_N(\xbf)
=
\frac{1}{Lh^d}
\sum_{j=1}^L
K\!\left(\frac{\xbf-\widetilde\ybf_j}{h}\right),
\end{equation}
where $K$ is a kernel (e.g.\ Gaussian) and $h>0$ is the
bandwidth. Equivalently, $\widehat f_N=K_h*\widehat\nu^N$,
where $K_h(\ubf)=h^{-d}K(\ubf/h)$ and $*$ denotes convolution.
Although the estimator uses $L$ predictive observations, the
subscript $N$ identifies the full-data construction. We obtain
$\widehat f_T$ analogously from the entire transformed coreset
trajectory $\{T(\widetilde\ybf_j^n)\}_{j=1}^L$.

Given an estimated density, we define density-based
clusters as the maximal path-connected components
of its upper level set. For $\widehat f_N$, these sets are 
$U_\tau(\widehat f_N)=\{\xbf\in\mathcal X:\widehat f_N(\xbf)\geq\tau\}$,
where $\mathcal X\subseteq\mathbb R^d$ is the domain
of interest and $\tau>0$ is a density threshold
\citep{Chacon2015,kriegel2011}. For numerical implementation,
we approximate these components using a radius-neighborhood
graph on trajectory observations whose estimated densities
exceed $\tau$ \citep{cuevas2001,ester96}.

To compare the full-data and coreset partitions at a common
clustering resolution, we calibrate the bandwidth and threshold
on the full-data trajectory in each paired replication, targeting
$K_{\mathrm{cl}}$ components (Section~\ref{app:implementation}
of the Supplementary Material). We then apply the same bandwidth,
threshold, and graph radius to the transformed coreset trajectory.
This shared calibration avoids differences from separately
chosen clustering resolutions. The target guides full-data
calibration but does not constrain the coreset component count
$K_T$; small components also count toward the target.

We define the resulting partitions on the reference set
$\mathcal R=\{\xbf_1,\ldots,\xbf_R\}$. The reference observations are excluded from DP conditioning
and coreset support selection, but are used to fit and select
the coreset transformation. For each trajectory, we evaluate
its density estimate on the reference set and assign labels
using the retained graph components. A reference point with
estimated density at least $\tau$ receives the label of the
component whose observations contribute most to its KDE;
points below $\tau$ receive distinct singleton labels. If no
component is retained, all reference points receive distinct
singleton labels. In this way, both trajectories induce
partitions of the same $R$ reference points.

We denote realizations of these random partitions by
$\widehat\sbf^N$ and $\widehat\sbf_T^n$, corresponding to
$S_N$ and $S_n(T)$, respectively. Their laws $\mu^N$ and
$\mu^n(T)$ are defined on $\mathcal S_R$, the set of partitions
of the $R$ reference points.

Because component labels are arbitrary, we compare partitions
using the variation of information metric (VoI) \citep{melia}.
For partitions $\sbf,\sbf'\in\mathcal S_R$, let $P(k,k')$
denote the proportion of reference points labeled
$k$ in $\sbf$ and $k'$ in $\sbf'$, with marginals
$P(k)$ and $P(k')$. Then
\[
\mathrm{VoI}(\sbf,\sbf')
=
H(\sbf\mid\sbf')+H(\sbf'\mid\sbf)
=
-\sum_{k,k'}P(k,k')
\log\frac{P(k,k')^2}{P(k)\,P(k')},
\]
where $H(\sbf\mid\sbf')$ is the conditional entropy
of $\sbf$ given $\sbf'$. VoI is invariant to label
permutations and bounded above by $\log R$ \citep{melia}.
We write $\mathrm{VoI}_{\mathrm{ref}}$ for this metric
under the singleton convention.

Using VoI as the ground metric, we define the
$1$-Wasserstein distance between the induced partition laws as
\[
J_W(T)=W_{\mathrm{VoI}}\!\left\{\mu^N,\mu^n(T)\right\},
\qquad T\in\mathcal T.
\]
Section~\ref{sec:fitting} develops a tractable paired
criterion that upper-bounds this distance.

\subsection{Fitting and selecting the transformation}
\label{sec:fitting}

Directly minimizing the Wasserstein objective $J_W(T)$
requires optimizing over couplings of partition laws.
We instead pair full-data and coreset predictive
trajectories to compare their partitions, then relax
the resulting loss to allow gradient-based fitting.
Let $\gamma_T$ denote the joint law of these paired
trajectories. Their induced partitions form a coupling
of $\mu^N$ and $\mu^n(T)$, giving
\[
J_{\mathrm{pair}}(T)
=
\mathbb E_{\gamma_T}\!\left[
\mathrm{VoI}_{\mathrm{ref}}(S_N,S_n(T))
\right],
\qquad
J_W(T)\leq J_{\mathrm{pair}}(T).
\]
In each replication, we draw a shared $\phi$ and
generate full-data and coreset predictive trajectories
independently conditional on $\phi$ and their respective
conditioning samples. We apply $T$ to the entire coreset
trajectory, including draws originating from $F_\phi$,
thereby transforming both the empirical and base-measure
contributions. Both trajectories are clustered using
settings calibrated from the full-data trajectory
(Section~\ref{sec:proxy}).
The expected loss averages over both trajectories and
$\pi_\phi$ under the specified DP predictive surrogate.
With $V$ independent paired replications, its Monte Carlo
estimator is
$\widehat J_{\mathrm{pair},V}(T)=V^{-1}\sum_{v=1}^{V}
\mathrm{VoI}_{\mathrm{ref}}(\widehat\sbf^{N,(v)},\widehat\sbf_T^{n,(v)})$.

The paired loss is not directly suitable for gradient-based
optimization because cluster assignments can change
discontinuously with $T$. We therefore replace two hard
labeling decisions with soft weights: whether a reference
point exceeds the density threshold and which component
receives it. The full-data reference labels remain fixed.
Let $C_1,\ldots,C_{K_T}\subseteq\{1,\ldots,L\}$
index the retained connected components of the
transformed coreset trajectory. The contribution of
component $k$ to the KDE at $\xbf$ is
\[
\widehat f_{T,k}(\xbf)
=
\frac{1}{Lh^d}
\sum_{j\in C_k}
K\!\left(
\frac{\xbf-T(\widetilde\ybf_j^n)}{h}
\right),
\qquad k=1,\ldots,K_T.
\]
Observations outside the retained components also
contribute to $\widehat f_T$, so
$\sum_k\widehat f_{T,k}\leq\widehat f_T$.
Under hard labeling, points with
$\widehat f_T(\xbf)\geq\tau$ receive the label of
the component with the largest KDE contribution;
otherwise, they form distinct singletons.

We replace the threshold decision with the smooth weight
\[
\sigma_\kappa(\xbf)
=
\left[
1+\exp\left\{
-\frac{\log\widehat f_T(\xbf)-\log\tau}{\kappa}
\right\}
\right]^{-1},
\]
where $\kappa>0$ controls the transition around $\tau$.
As $\kappa\downarrow0$, this weight approaches the
hard threshold indicator when
$\widehat f_T(\xbf)\neq\tau$.
For $K_T\geq1$, we similarly relax component assignment:
\[
w_{k,\beta}(\xbf)
=
\frac{\widehat f_{T,k}(\xbf)^{1/\beta}}
{\sum_{k'=1}^{K_T}
\widehat f_{T,k'}(\xbf)^{1/\beta}},
\qquad k=1,\ldots,K_T,
\]
where $\beta>0$ controls concentration toward
components with larger KDE contributions.
The resulting soft component weights are
$q_k(\xbf)=\sigma_\kappa(\xbf)w_{k,\beta}(\xbf)$,
for $k=1,\ldots,K_T$.
The remaining weight $1-\sigma_\kappa(\xbf)$
belongs to the reference point's own singleton block.
If $K_T=0$, every reference point has singleton weight one.

These soft weights and the fixed full-data labels
define a contingency table over full-data clusters,
coreset components, and the $R$ distinct singleton
blocks. Its conditional entropies yield the relaxed
paired loss $\ell_{\kappa,\beta}(T)$, with the table
construction detailed in Section~\ref{app:implementation}
of the Supplementary Material.
The population loss is
$\mathcal L_{\kappa,\beta}(T)=
\mathbb E_{\gamma_T}[\ell_{\kappa,\beta}(T)]$,
with Monte Carlo version
$\widehat{\mathcal L}_{\kappa,\beta,V}(T)
=V^{-1}\sum_{v=1}^{V}\ell_{\kappa,\beta}^{(v)}(T)$.
For fixed graph component memberships, the relaxed
loss permits differentiation with respect to $T$.
We periodically recompute components, holding their
memberships fixed during gradient steps. Thus, fitting
uses local gradients without requiring global
differentiability, while final selection evaluates
the original hard paired loss on fresh simulations.

We fit an affine map $T(\xbf)=\Abf\xbf+\bbf$,
starting from the identity, using Adam
\citep{kingma2015adam}.
We generate $V_{\mathrm{col}}$ paired trajectories
once and use mini-batches of these draws and reference
points for gradient updates. The smoothing parameters
$\kappa$ and $\beta$ decrease geometrically.
We constrain fitting to a compact family
$\mathcal K\subseteq\mathcal T$, with singular
values of $\Abf$ in $[1/c,c]$ and
$\|\bbf\|\leq b_{\max}$, enforcing these bounds
after each update.
Among periodically recorded iterates, we retain
$\bar T$ with the smallest average hard paired VoI
over the complete collection of draws.
The bounds and optimization settings are discussed
in Section~\ref{sec:settings} and
Section~\ref{app:implementation} of the Supplementary
Material.

To select the final transformation, we evaluate
a finite path between the identity and the fitted map:
$T_\lambda=I+\lambda(\bar T-I)$, with $\lambda$
ranging over a finite grid $\Lambda\subseteq[0,1]$.
Throughout this work, we use
$\Lambda=\{0,0.1,\ldots,1\}$. The identity ($\lambda=0$) corresponds to uniform
subsampling, while intermediate values allow smaller
adjustments than $\bar T$.
Conditional on fitting, the candidate set
$\mathcal A_G=\{T_\lambda:\lambda\in\Lambda\}$,
with $G=|\Lambda|=11$, is fixed before selection.
We do not project candidates onto $\mathcal K$,
which is not generally convex. The candidates remain
in the affine family $\mathcal T$, and selection
operates on transformed trajectories without requiring
invertible maps.
Using $V_{\mathrm{sel}}$ fresh paired replications,
we choose
$\widehat\lambda\in\argmin_{\lambda\in\Lambda}
\widehat J_{\mathrm{pair},V_{\mathrm{sel}}}(T_\lambda)$,
resolving ties in favor of smaller $\lambda$.
We use $V_{\mathrm{eval}}$ additional, independent
replications to compare
$\widehat T=T_{\widehat\lambda}$ with the identity.
The final coreset is $(\ys_{1:n},\widehat T)$.

Because the candidate path includes the identity,
selection cannot increase the empirical paired loss
relative to uniform subsampling on the selection
replications. Proposition~\ref{prop:selection_control}
provides the corresponding population guarantee,
up to finite-simulation error, without requiring
a globally optimal fit.

Algorithm~\ref{alg:coreset} summarizes the construction.

\begin{algorithm}[!htbp]
\caption{Coreset for partitions}
\label{alg:coreset}
\begin{algorithmic}[1]
\Require Data $\ybf_{1:N}$; sizes $R,n$;
urn specification $(\alpha,F_\phi,\pi_\phi)$;
trajectory length $L$; target $K_{\mathrm{cl}}$.
\Ensure Coreset $(\ys_{1:n},\widehat T)$.
\State Split the data into reference and analysis sets
$\mathcal R,A$; sample $\ys_{1:n}$ uniformly from $A$.
\State For $v=1,\ldots,V_{\mathrm{col}}$, draw
$\phi_v\sim\pi_\phi$ and conditionally independent
$G_{A,v}\sim\mathrm{DP}(\alpha+|A|,F_{\phi_v}^{A})$ and
$G_{\star,v}\sim\mathrm{DP}(\alpha+n,F_{\phi_v}^{\star})$.
Generate length-$L$ trajectories and label $\mathcal R$
using full-trajectory clustering settings.
\State Starting from $I$, optimize
$\mathcal L_{\kappa,\beta}$ over $\mathcal K$,
decreasing $\kappa,\beta$ and periodically updating
graph components. Retain
$\bar T=\arg\min_{T\in\{T^{(s)}\}}
\widehat J_{\mathrm{pair},V_{\mathrm{col}}}(T)$
over fitting checkpoints.
\State Set $T_\lambda=I+\lambda(\bar T-I)$ for
$\lambda\in\Lambda=\{0,0.1,\ldots,1\}$.
\State Using $V_{\mathrm{sel}}$ fresh pairs, select
$\widehat\lambda=\arg\min_{\lambda\in\Lambda}
\widehat J_{\mathrm{pair},V_{\mathrm{sel}}}(T_\lambda)$,
breaking ties toward smaller $\lambda$.
\State \Return $(\ys_{1:n},T_{\widehat\lambda})$.
\end{algorithmic}
\end{algorithm}

\section{Theoretical analysis}
\label{sec:theory}

The proposed method fits a relaxed loss and selects a
transformation using independent predictive simulations.
We now establish how these steps relate to the
partition-law objective, providing bounds for the
relaxation and selection procedures. We also study
conditions under which convergence of smoothed predictive
densities yields agreement of their induced partitions.
For consistency, we consider ideal density-based clustering
at a fixed resolution. Throughout, we condition on the
reference set \(\mathcal R\), with \(R\geq2\), and use
the coupling \(\gamma_T\) from Section~\ref{sec:fitting}.

We state the conditions used by the three results here;
each condition applies only where cited.

\noindent\textbf{(A1) Relaxation margins.}
For a fixed \(T\), compute soft and hard assignments
using the same density and component scores, with
graph components held fixed. Assume a strictly
positive kernel and singleton labels below \(\tau\).
Almost surely under \(\gamma_T\), every reference point
\(\xbf_r\) satisfies \(\widehat f_T(\xbf_r)\neq\tau\).
If \(\widehat f_T(\xbf_r)>\tau\) and \(K_T\geq2\),
then \(\arg\max_{1\leq k\leq K_T}
\widehat f_{T,k}(\xbf_r)\) is unique, where \(K_T\)
is the number of retained coreset components.

\noindent\textbf{(A2) Independent finite-grid selection.}
Condition on fitting and a finite candidate set
\(\mathcal A_G\subseteq\mathcal T\) of size \(G\)
containing \(I\). The \(V\) selection replications are
conditionally independent and identically distributed,
with hard VoI evaluated on the reference set.

\noindent\textbf{(A3) Predictive-density regularity.}
Let \(\Ybf_1,\Ybf_2,\ldots\) be i.i.d.\ from a probability
measure \(P_0\) on \(\mathbb R^d\), with compact
\(\mathcal X\subset\mathbb R^d\). Fix \(\alpha>0\) and
\(h>0\). For each conditioning sample of size \(m\),
let \(\pi_m(\dd\phi\mid\Ybf_{1:m})\) be a probability
kernel on parameters \(\phi\) with
\(F_\phi\in\mathcal P(\mathbb R^d)\). This mixing
distribution may depend on the observations through
empirical calibration. Assume \(K_h\) is bounded and
Lipschitz in its first argument on \(\mathcal X\),
uniformly over its second argument, as in
Lemma~\ref{lem:urn_kde_concentration}.

\noindent\textbf{(A4) Partition stability.}
Fix reference points \(\xbf_1,\ldots,\xbf_R\in\mathcal X\)
and a threshold \(\tau\). Suppose \(g_0=K_h*P_0\)
satisfies the persistent-core and valley-separation
conditions of Definition~\ref{def:persistent_core}, and
\(g_0(\xbf_r)\ne\tau\) for each reference point.
Use the same ideal path-connected upper-level-set partition
map for all predictive laws, with either a common noise
label or reference-indexed singleton labels below the
threshold. The full-data and identity-coreset conditioning
sizes \(N_u,n_u\) and predictive trajectory lengths
\(L_u\) all diverge.

\noindent\textbf{(A5) Asymptotic selection.}
Conditional on fitting, the \(G_u\) selection candidates
are fixed and include \(I\), with \(V_u\) independent
fresh paired draws. For fixed \(R\), suppose
\(V_u\to\infty\) and \(\log(G_u)/V_u\to0\).

\subsection{Control of the relaxed loss}
\label{sec:theory_relax}

The following proposition bounds the Wasserstein--VoI
distance by the relaxed fitting loss, with an error
determined by density-threshold margins and the separation
of competing component scores.

\begin{proposition}[Control by the relaxed loss]
\label{prop:relaxation_control}
Under (A1), let \(e_{\kappa,\beta}\) be the threshold- and
component-margin error defined in
\eqref{eq:relax_error} of the Supplementary Material,
with \(e_{\kappa,\beta}=0\) if \(K_T=0\).
Set \(D_T=R(K_T+R)\), and let \(\omega_{D_T}\) be the
entropy-continuity modulus defined in
Section~\ref{app:relax_selection} of the Supplementary
Material. Then
\[
J_W(T)
\le
\mathcal L_{\kappa,\beta}(T)
+
3\mathbb E_{\gamma_T}\!\left[
\omega_{D_T}(e_{\kappa,\beta})\right].
\]
Moreover, for fixed \(T\),
\(\mathcal L_{\kappa,\beta}(T)
\to J_{\mathrm{pair}}(T)\) as
\(\kappa,\beta\downarrow0\).
\end{proposition}

The proposition justifies our fitting surrogate:
the relaxed loss controls the Wasserstein--VoI distance
up to an explicit approximation error. During fitting,
we decrease $\kappa$ and $\beta$ geometrically. This
error vanishes as both parameters approach zero, provided
no reference point lies exactly at the density threshold
or has tied leading component scores. The bound applies
to each realized graph decomposition, including those
updated during fitting. The proof is in the Supplementary
Material.

\subsection{Finite-grid selection}
\label{sec:theory_select}

Fitting with the relaxed loss does not guarantee improved
partition agreement over uniform subsampling. We therefore
select from a finite set of transformations that includes
the identity, using hard paired VoI on fresh predictive
replications. The following proposition shows that this
selection controls population paired loss up to a
finite-simulation error, without requiring the candidates
to belong to the generally nonconvex fitting family
$\mathcal K$.

\begin{proposition}[Finite-grid selection]
\label{prop:selection_control}
Under (A2), let \(\widehat T\) minimize the empirical paired
loss over \(\mathcal A_G\). For \(0<\delta<1\), with conditional
probability at least \(1-\delta\),
\[
\begin{aligned}
J_W(\widehat T)
&\le J_{\mathrm{pair}}(\widehat T)\\
&\le \min_{T\in\mathcal A_G}J_{\mathrm{pair}}(T)
+2\varepsilon_{V,G,\delta}\\
&\le J_{\mathrm{pair}}(I)+2\varepsilon_{V,G,\delta},
\end{aligned}
\]
where
\(\varepsilon_{V,G,\delta}=
\log R\sqrt{\log(2G/\delta)/(2V)}\).
If a candidate improves on \(I\) by more than
\(2\varepsilon_{V,G,\delta}\) in population paired loss,
then so does \(\widehat T\) on the same event.
Its empirical selection loss is at most that of \(I\)
exactly.
\end{proposition}

The result justifies treating fitting as a way to propose
transformations rather than committing to its output.
The selection step can retain the uniform subsample
when the candidates fail to improve its empirical paired
loss. The proof is in the Supplementary Material.

\subsection{Stability of predictive partitions}
\label{sec:theory_stab}

We next study whether the coreset preserves predictive
partition laws asymptotically. Since density convergence
does not ensure partition agreement, we adapt the level-set
stability argument of \citet{bariletto2026} under the
persistent-core and valley-separation conditions (A4),
without requiring exact recovery of the component count.
The theorem combines this stability with smoothed DP
contraction, conditional KDE concentration, and finite-grid
selection.

\begin{theorem}[Consistency of predictive partition matching and selection]
\label{thm:selected_conditional}
Under (A3)--(A4), the full-data and identity-coreset
reference partitions converge in probability to the same
deterministic partition \(S_0\). Moreover, in probability
over the conditioning samples,
\[
J_{\mathrm{pair}}(I)\xrightarrow{\Pr}0,
\qquad J_W(I)\xrightarrow{\Pr}0.
\]
If (A5) also holds, then
\begin{equation}\label{eq:selected_consistency}
J_W(\widehat T_u)\xrightarrow{\Pr}0.
\end{equation}
\end{theorem}

The theorem establishes that uniform subsampling is
asymptotically sufficient for matching ideal predictive
partitions and that selecting a transformation preserves
this consistency. The result concerns fixed-bandwidth,
fixed-threshold partitions of the smoothed distribution
\(g_0=K_h*P_0\), not the unsmoothed density or the
adaptive graph implementation. Graph-to-level-set
approximation and the additional selection error are
addressed in Proposition~\ref{prop:graph_selection}
and Lemma~\ref{lem:graph_levelset} of the Supplementary
Material, which also contains the theorem's proof.

\section{Empirical validation}
\label{sec:experiments}

We now investigate whether the proposed coreset preserves
full-data predictive clustering in finite samples while
reducing computational costs. We compare it with uniform
subsampling and alternative coresets using simulated mixtures
and mass cytometry data, examining performance across analysis
settings. An additional curved-cluster simulation appears in
Section~\ref{app:extra_studies} of the Supplementary Material.

For each dataset replicate, all methods start from the same
uniformly sampled support. Within each pair of predictive
trajectories, we calibrate the bandwidth and threshold on
the full-data trajectory to target \(K_{\mathrm{cl}}\)
connected components, then reuse these settings for the
reduced trajectory. This ensures a common clustering resolution
\citep{stuetzle2003}. The reference set is excluded from
predictive conditioning but used to fit and select our
transformation; it remains fixed during evaluation.

We measure agreement using mean hard VoI over
\(V_{\mathrm{eval}}=50\) fresh paired predictive trajectories,
independent of fitting and selection. Each pair shares a draw
of \(\phi\), with trajectories generated separately.
Denoting the mean by \(\widehat J_{\mathrm{eval}}(T)\),
we report the relative improvement over uniform subsampling,
$\{\widehat J_{\mathrm{eval}}(I)-\widehat J_{\mathrm{eval}}(\widehat T)\}/\widehat J_{\mathrm{eval}}(I)$,
and paired 95\% \(t\)-intervals for the mean loss difference.
These intervals capture predictive Monte Carlo variation
conditional on the dataset, reference set, support, and
selected transformation, rather than variation across datasets.
The evaluation estimates paired VoI, not the partition-law
Wasserstein distance \(J_W\). In simulations, we also
report the adjusted Rand index (ARI) \citep{hubert1985}
against generating mixture labels. Neither these labels
nor ARI are used in coreset construction, so this
comparison provides an external assessment of group
recovery, distinct from the predictive partition agreement
targeted by our optimization.

\subsection{Choice of settings}
\label{sec:settings}

We use the same predictive and fitting settings in the two main
studies, except for clustering resolution. Specifically, we take
$\alpha=1$, predictive trajectory length $L=500$, and
reference size $R=2000$. The reference size determines the number of
observations on which partition agreement is evaluated.
At $R=2000$, a group with 1\% prevalence contributes
approximately 20 reference observations in expectation.
We set $K_{\mathrm{cl}}=10$ in the simulation, matching the
number of generating mixture components, and
$K_{\mathrm{cl}}=24$ in cytometry, matching the number
of manually gated populations. This parameter determines
the target resolution of full-data clustering and is not
optimized by the coreset. We use singular-value bounds
$[1/c,c]$ with $c=5$ and a translation bound
$b_{\max}=3s_{\mathrm{traj}}$, where $s_{\mathrm{traj}}$ is the root-mean-square
distance of pooled full-data trajectory observations
from their mean. Further implementation details and
sensitivity analyses appear in Sections~\ref{app:implementation},
\ref{app:kcl}, and~\ref{app:alpha} of the Supplementary Material.

\subsection{Simulation}
\label{sec:sims}

To examine how component separation and rare-group
prevalence affect predictive partition agreement, we
generate \(N=20{,}000\) observations in \(d=8\)
dimensions from a ten-component Gaussian mixture.
Component means are drawn from
\(\mathcal N(\mathbf 0,\Delta^2\Ibf_d)\), where
\(\Delta\in\{2,3\}\) controls their typical separation.
Components have diagonal standard deviations in
\([0.5,1.5]\), and two components have prevalence
\(p_{\min}\in\{0.05,0.01\}\).
We consider \(n\in\{50,100,250\}\) in all twelve settings,
with 20 replicates each. Each uses \(V_{\mathrm{col}}=40\)
collection draws and 500 optimization steps, followed by
50 of 100 fresh pairs for selection and 50 independent
pairs for evaluation.

The partition coreset improves agreement with the full-data
predictive clustering in every setting (Figure~\ref{fig:main};
Supplementary Table~\ref{tab:valleys}). Median paired-loss reductions
range from 13\% to 30\%, and the paired 95\% interval
favors the coreset in 236 of 240 replicates, including at
least 19 of 20 replicates in each setting. Reductions are
20--30\% at \(n=50\) and 13--16\% at \(n=250\).
The improvements also persist when two mixture components
have only 1\% prevalence.

Recovery of the generating mixture groups gives a
complementary assessment. For the more separated mixtures
(\(\Delta=3\)), the coreset has higher mean ARI than uniform
subsampling in every setting, with advantages of
approximately 0.03--0.13. For \(\Delta=2\), the coreset
continues to reduce paired VoI, but ARI improvements are
smaller and disappear or reverse at \(n=250\). Thus,
predictive partition agreement can improve even when the
density-based groups do not coincide with the generating
mixture labels.

\begin{figure}[!htbp]
\centering
\includegraphics[width=0.95\textwidth]{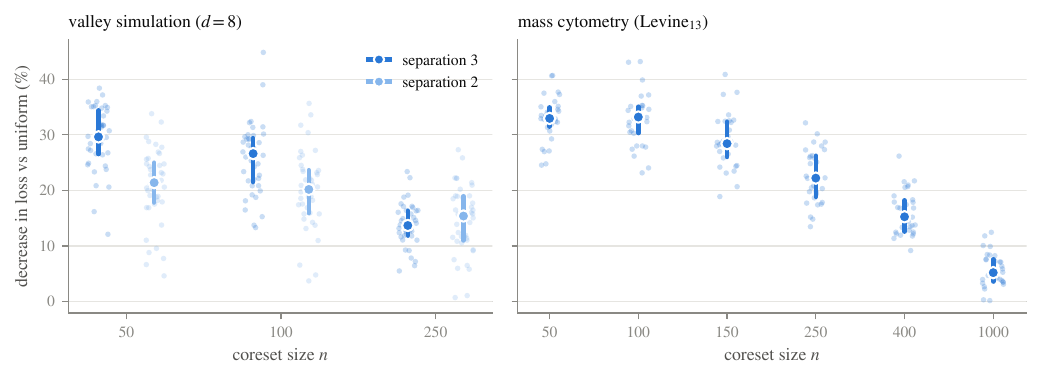}
\caption{Relative decrease in $\widehat J_{\mathrm{eval}}$ for the partition coreset versus uniform subsampling on the same support in the valley-separated simulation and Levine$_{13}$ cytometry. Faint points show individual replicates; large points and bars show medians and interquartile ranges. Positive values favor the coreset.}
\label{fig:main}
\end{figure}

\subsection{Mass cytometry}
\label{sec:realdata}

To assess predictive partition agreement among cell
populations of varying prevalence and potentially weak
separation, we use the Levine$_{13}$ mass cytometry dataset
\citep{levine2015,weber2016}, containing 167{,}044 human
bone-marrow cells measured on 13 protein markers.
We apply the customary $\mathrm{arcsinh}(x/5)$
transformation to the marker intensities and set
$K_{\mathrm{cl}}=24$, matching the number of manually
gated populations without using their labels for clustering.

We use the first half of the dataset to develop the method
and select its settings, and the second half
($N=83{,}522$) for evaluation. We consider support sizes
$n\in\{50,100,150,250,400,1000\}$, with 30 independent
replicates per size. Each replicate uses
$V_{\mathrm{col}}=40$ collection draws and 500
optimization steps; the remaining evaluation settings
are as described above.

The partition coreset improves paired VoI over uniform
subsampling at every support size (Figure~\ref{fig:main}).
The gains are largest for smaller coresets and gradually
decrease as the support grows, while paired 95\% intervals
favor the coreset in nearly all replicates. The Supplementary Material provides detailed numerical
results, comparisons with larger uniform subsamples,
and representative clustering visualizations.

\subsection{Alternative coresets and computational performance}

We compare our method with coresets approximating
likelihoods or empirical distributions, since preserving
these quantities can also preserve density-based clustering.
The quasi-Newton coreset (QNC) of \citet{naik2022fast}
optimizes support weights to approximate the likelihood
of a fitted Gaussian mixture. The \(W_2\) coreset based on
\citet{claici_wasserstein_2020} fits an affine map by
matching the empirical distribution to a sample of
5{,}000 observations. The implementations also differ:
QNC weights are not normalized to sum to \(n\), changing
the effective prior mass in the predictive urn, while
the \(W_2\) baseline uses different affine constraints
and omits selection along the path to the identity.
Thus, the comparisons evaluate complete coreset
constructions rather than isolate their objectives
(Supplementary Sections~\ref{app:qnc} and~\ref{app:w2}).

The results show that QNC produces larger paired losses
than uniform subsampling across the reported settings
(Figure~\ref{fig:baselines}). Its median losses exceed
those of uniform subsampling by 6--24\% in the simulation
and 10--26\% for Levine$_{13}$.

The \(W_2\) baseline is more competitive, although its
performance varies across settings. In the simulation
with \(\Delta=2\), it improves over uniform subsampling
by a median of 11--16\%, but its gains are less consistent
at \(\Delta=3\). For Levine$_{13}$ at \(n\le150\),
the partition coreset has median paired losses 3--4\%
higher than \(W_2\). However, the comparison reverses
at larger support sizes: for \(n\ge400\), intervals favor
the partition coreset in 56 of 60 replicates. These
results suggest that observation-space matching can be
effective for very small supports, while partition-targeted
fitting provides more consistent improvements at larger
support sizes in the cytometry study.

\begin{figure}[!htbp]
\centering
\includegraphics[width=0.95\textwidth]{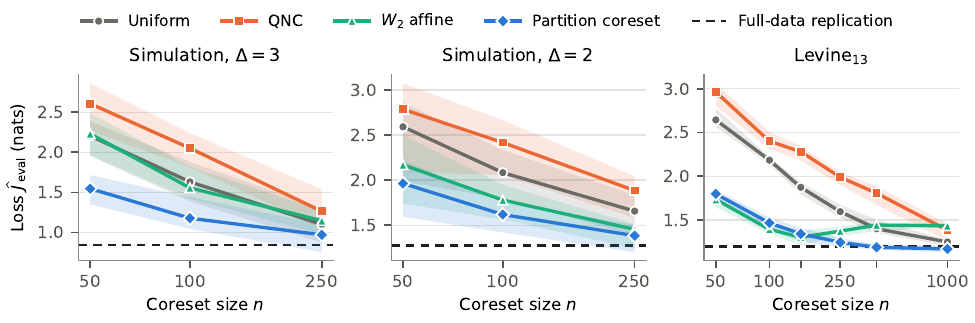}
\caption{Paired evaluation loss $\widehat J_{\mathrm{eval}}$
against support size for uniform subsampling, QNC, the
observation-space $W_2$ map, and the partition coreset.
Lines and bands show medians and interquartile ranges
across replicates (40 per size in the pooled simulation
panels and 30 for Levine$_{13}$). Dashed lines show the
median loss between independent full-data trajectories
conditional on a shared $\phi$
(Section~\ref{app:replication} of the Supplementary Material).}
\label{fig:baselines}
\end{figure}

The construction is computationally inexpensive in the
settings considered, despite requiring predictive
simulations and gradient-based optimization. The median
total construction time, including predictive trajectory
collection, transformation fitting, and selection, is
51 seconds in the main simulation and 50 seconds for
Levine$_{13}$. All constructions run on a single core
of an Apple M5 Pro laptop. The
corresponding median fitting times for \(W_2\) are
3 and 7 seconds, while median QNC construction times
are 7 and 68 seconds, respectively. These timings
cover different operations: the \(W_2\) values measure
fitting, whereas the QNC values describe its own
construction procedure.

Additional experiments in the Supplementary Material
examine curved-cluster geometries, clustering resolution,
and sensitivity to the DP concentration parameter
(Sections~\ref{app:extra_studies}, \ref{app:kcl},
and \ref{app:alpha}). Each clustering resolution
defines a different full-data predictive partition target.

\section{Discussion}\label{sec:discussion}

We have developed a coreset method that targets the probability law of
predictive density-based partitions, rather than empirical distributions,
likelihoods, or parameter posteriors. Our conceptual objective is the
Wasserstein distance between partition laws, using variation of information
as the ground metric. We fit an affine transformation through a relaxation
of paired expected VoI, which upper-bounds this distance under the
specified predictive coupling.

For fixed graph component memberships, the relaxed loss approaches hard
paired VoI as smoothing vanishes, while fresh replications provide
finite-grid selection guarantees without requiring globally optimal
fitting. The consistency result concerns ideal fixed-resolution
level-set partitions under geometric stability conditions,
not the full adaptive graph procedure.

Across the mixture, cytometry, and curved-cluster studies,
transformed coresets improve paired predictive partition
agreement over uniform subsampling, with modest
construction times. In simulations, we also assess
recovery of generating groups through ARI, which is
not used in fitting or selection. The improvements in
ARI, particularly for well-separated mixtures, suggest
that gains in predictive partition agreement can
translate into better clustering accuracy, although
this is not guaranteed. The observation-space \(W_2\)
baseline remains competitive for the smallest cytometry
supports, and the comparisons do not isolate optimization
targets from other methodological differences. Moreover,
lower empirical paired VoI does not necessarily imply
the same Wasserstein ranking across methods.

An important practical advantage is that the coreset is
inexpensive to construct, requiring approximately 50 seconds
in our main experiments. Moreover, because the selection
path includes the identity transformation, the method
cannot perform worse than uniform subsampling on the
empirical selection criterion. The finite-grid guarantee
extends this safeguard to population paired loss, up to
Monte Carlo error. Thus, the method offers an inexpensive
opportunity to improve partition agreement while retaining
uniform subsampling as a fallback, making it attractive
even when the potential gains are modest.

The affine family provides a flexible and tractable
way to relocate coreset support points. More expressive
transformations, such as normalizing flows, could better
adapt the support to complex cluster geometries.
Other extensions include combining reweighting with
relocation and improving support selection, for example
through \(k\)-center initialization \citep{jang2019}.
 Clustering across density thresholds
\citep{stuetzle2003} could extend the fixed-resolution
approach. The same objective could also be investigated in Bayesian
model-based clustering, where posterior mixture allocations induce
random partitions.
More broadly, our results illustrate how data reduction can
target uncertainty in the outcome of a statistical analysis, rather
than only approximating the underlying data distribution.

\section*{Supplementary Material}
The online Supplementary Material contains proofs and further theoretical
results, implementation details, baseline specifications, a curved-cluster
simulation, and sensitivity studies.

\section*{Acknowledgments}
The author thanks Peter Mueller, Nhat Ho, Arya Farahi, and Nicola Bariletto for helpful suggestions during the development of this work.
\paragraph{Funding and competing interests.} No funding was received for this study, and the author has no competing interests to declare.

\paragraph{Data availability.} The simulated data are generated as described; the Levine$_{13}$ dataset is publicly available (Section~\ref{sec:realdata}).

\paragraph{Code availability.} Python code for reproducing the numerical experiments is available at\newline\url{https://github.com/BernardoFL/density-clustering-coresets}.

\bibliographystyle{plainnat}
\bibliography{main}

@article{hubert1985,
  author  = {Hubert, Lawrence and Arabie, Phipps},
  title   = {Comparing partitions},
  journal = {Journal of Classification},
  volume  = {2},
  pages   = {193--218},
  year    = {1985},
  doi     = {10.1007/BF01908075}
}

@article{levine2015,
  author  = {Levine, Jacob H. and Simonds, Erin F. and Bendall, Sean C. and Davis, Kara L. and Amir, El-ad D. and Tadmor, Michelle and Litvin, Oren and Fienberg, Harris and Jager, Astraea and Zunder, Eli and Finck, Rachel and Gedman, Amanda L. and Radtke, Ina and Downing, James R. and Pe'er, Dana and Nolan, Garry P.},
  title   = {Data-Driven Phenotypic Dissection of {AML} Reveals Progenitor-like Cells that Correlate with Prognosis},
  journal = {Cell},
  year    = {2015},
  volume  = {162},
  number  = {1},
  pages   = {184--197},
  doi     = {10.1016/j.cell.2015.05.047}
}

@article{weber2016,
  author  = {Weber, Lukas M. and Robinson, Mark D.},
  title   = {Comparison of clustering methods for high-dimensional single-cell flow and mass cytometry data},
  journal = {Cytometry Part A},
  year    = {2016},
  volume  = {89},
  number  = {12},
  pages   = {1084--1096},
  doi     = {10.1002/cyto.a.23030}
}

@article{stuetzle2003,
  author  = {Stuetzle, Werner},
  title   = {Estimating the Cluster Tree of a Density by Analyzing the Minimal Spanning Tree of a Sample},
  journal = {Journal of Classification},
  year    = {2003},
  volume  = {20},
  number  = {1},
  pages   = {25--47},
  doi     = {10.1007/s00357-003-0004-6}
}

@inproceedings{ester96,
author = {Ester, Martin and Kriegel, Hans-Peter and Sander, J\"{o}rg and Xu, Xiaowei},
title = {A density-based algorithm for discovering clusters in large spatial databases with noise},
year = {1996},
publisher = {AAAI Press},
booktitle = {Proceedings of the Second International Conference on Knowledge Discovery and Data Mining},
pages = {226--231},
numpages = {6},
location = {Portland, Oregon},
series = {KDD'96}
}

@article{Chacon2015,
  author  = {Chac{\'o}n, Jos{\'e} E.},
  title   = {A Population Background for Nonparametric Density-Based Clustering},
  journal = {Statistical Science},
  year    = {2015},
  volume  = {30},
  number  = {4},
  pages   = {518--532},
  doi     = {10.1214/15-STS526}
}

@article{kriegel2011,
author = {Kriegel, Hans-Peter and Kr{\"o}ger, Peer and Sander, J{\"o}rg and Zimek, Arthur},
title = {Density-based clustering},
journal = {WIREs Data Mining and Knowledge Discovery},
volume = {1},
number = {3},
pages = {231--240},
doi = {10.1002/widm.30},
year = {2011},
}

@inproceedings{
naik2022fast,
title={Fast {B}ayesian Coresets via Subsampling and Quasi-{N}ewton Refinement},
author={Cian Vasanttilak Naik and Judith Rousseau and Trevor Campbell},
booktitle={Advances in Neural Information Processing Systems},
editor={Alice H. Oh and Alekh Agarwal and Danielle Belgrave and Kyunghyun Cho},
year={2022},
url={https://openreview.net/forum?id=M3WW7TqoMvc},
  doi = {10.52202/068431-0006}
}

@article{rubin_bayesian_1981,
  author    = {Donald B. Rubin},
  title     = {The {Bayesian} Bootstrap},
  journal   = {The Annals of Statistics},
  year      = {1981},
  volume    = {9},
  number    = {1},
  pages     = {130--134},
  doi       = {10.1214/aos/1176345338},
}

@article{blackwell1973,
author = {David Blackwell and James B. MacQueen},
title = {{Ferguson Distributions Via {P}olya Urn Schemes}},
volume = {1},
journal = {The Annals of Statistics},
number = {2},
publisher = {Institute of Mathematical Statistics},
pages = {353 -- 355},
year = {1973},
doi = {10.1214/aos/1176342372},
}

@article{fong2024,
    author = {Fong, Edwin and Holmes, Chris and Walker, Stephen G},
    title = {Martingale posterior distributions},
    journal = {Journal of the Royal Statistical Society Series B: Statistical Methodology},
    volume = {85},
    number = {5},
    pages = {1357--1391},
    year = {2023},
    doi = {10.1093/jrsssb/qkad005},
}

@book{villani2003topics,
  title={Topics in Optimal Transportation},
  author={Villani, C.},
  series = {Graduate Studies in Mathematics},
  volume = {58},
  doi = {10.1090/gsm/058},
  year={2003},
  publisher={American Mathematical Society}
}

@inproceedings{manousakas2020,
 author = {Manousakas, Dionysis and Xu, Zuheng and Mascolo, Cecilia and Campbell, Trevor},
 booktitle = {Advances in Neural Information Processing Systems},
 editor = {H. Larochelle and M. Ranzato and R. Hadsell and M.F. Balcan and H. Lin},
 pages = {14950--14960},
 publisher = {Curran Associates, Inc.},
 title = {Bayesian Pseudocoresets},
 volume = {33},
 year = {2020}
}

@inproceedings{huggins2016,
author = {Huggins, Jonathan H. and Campbell, Trevor and Broderick, Tamara},
title = {Coresets for scalable {B}ayesian logistic regression},
year = {2016},
publisher = {Curran Associates Inc.},
booktitle = {Proceedings of the 30th International Conference on Neural Information Processing Systems},
pages = {4087--4095},
numpages = {9},
location = {Barcelona, Spain},
series = {NIPS'16}
}

@article{melia,
title = {Comparing clusterings—an information based distance},
journal = {Journal of Multivariate Analysis},
volume = {98},
number = {5},
pages = {873--895},
year = {2007},
author = {Marina Meil{\u a}},
  doi = {10.1016/j.jmva.2006.11.013}
}

@inproceedings{lyddon_nonparametric_2018,
	title = {Nonparametric learning from {Bayesian} models with randomized objective functions},
	volume = {31},
	urldate = {2024-02-22},
	booktitle = {Advances in {Neural} {Information} {Processing} {Systems}},
	publisher = {Curran Associates, Inc.},
	author = {Lyddon, Simon and Walker, Stephen and Holmes, Chris C},
	year = {2018},
}

@article{claici_wasserstein_2020,
	title   = {Wasserstein {Measure} {Coresets}},
	author  = {Claici, Sebastian and Genevay, Aude and Solomon, Justin},
	journal = {arXiv preprint arXiv:1805.07412},
	year    = {2020},
	url     = {https://arxiv.org/abs/1805.07412}
}

@article{campbell_automated_2019,
	title = {Automated {Scalable} {Bayesian} {Inference} via {Hilbert} {Coresets}},
	volume = {20},
	number = {15},
	urldate = {2024-02-28},
	journal = {Journal of Machine Learning Research},
	author = {Campbell, Trevor and Broderick, Tamara},
	year = {2019},
	pages = {1--38},
}

@article{ferguson_bayesian_1973,
    title = {A {Bayesian} {Analysis} of {Some} {Nonparametric} {Problems}},
    volume = {1},
    issn = {0090-5364, 2168-8966},
    doi = {10.1214/aos/1176342360},
    number = {2},
    urldate = {2023-09-26},
    journal = {The Annals of Statistics},
    author = {Ferguson, Thomas S.},
    year = {1973},
    pages = {209--230},
}

@article{rinaldo_stability_2010,
    title = {Generalized density clustering},
    author = {Rinaldo, Alessandro and Wasserman, Larry},
    journal = {The Annals of Statistics},
    volume = {38},
    number = {5},
    pages = {2678--2722},
    year = {2010},
    publisher = {Institute of Mathematical Statistics},
    doi = {10.1214/10-AOS797},
}

@article{bariletto2026,
    title = {Scalable Posterior Uncertainty for Flexible Density-Based Clustering},
    author = {Bariletto, Nicola and Walker, Stephen G.},
    journal = {arXiv preprint arXiv:2603.03188},
    year = {2026},
    url = {https://arxiv.org/abs/2603.03188},
}

@inproceedings{kingma2015adam,
    author    = {Kingma, Diederik P. and Ba, Jimmy},
    title     = {Adam: {A} Method for Stochastic Optimization},
    booktitle = {International Conference on Learning Representations},
    year      = {2015},
}

@article{wade2018,
  author  = {Wade, Sara and Ghahramani, Zoubin},
  title   = {Bayesian Cluster Analysis: Point Estimation and Credible Balls (with Discussion)},
  journal = {Bayesian Analysis},
  year    = {2018},
  volume  = {13},
  number  = {2},
  pages   = {559--626}
}

@inproceedings{chen2022shf,
  author    = {Chen, Naitong and Xu, Zuheng and Campbell, Trevor},
  title     = {{B}ayesian Inference via Sparse {H}amiltonian Flows},
  booktitle = {Advances in Neural Information Processing Systems},
  volume    = {35},
  year      = {2022},
  doi = {10.52202/068431-1518}
}

@article{buch2026,
  author = {Buch, David and Dewaskar, Miheer and Dunson, David B.},
  title = {{B}ayesian Level Set Clustering},
  journal = {Journal of Machine Learning Research},
  volume = {27},
  number = {159},
  pages = {1--68},
  year = {2026}
}

@article{miller2014,
  author = {Miller, Jeffrey W. and Harrison, Matthew T.},
  title = {Inconsistency of {P}itman--{Y}or Process Mixtures for the Number of Components},
  journal = {Journal of Machine Learning Research},
  volume = {15},
  pages = {3333--3370},
  year = {2014}
}

@article{phillips2020,
  author = {Phillips, Jeff M. and Tai, Wai Ming},
  title = {Near-Optimal Coresets of Kernel Density Estimates},
  journal = {Discrete \& Computational Geometry},
  volume = {63},
  number = {4},
  pages = {867--887},
  year = {2020},
  doi = {10.1007/s00454-019-00134-6}
}

@inproceedings{bachem2015,
  author = {Bachem, Olivier and Lucic, Mario and Krause, Andreas},
  title = {Coresets for Nonparametric Estimation -- the Case of {DP}-Means},
  booktitle = {Proceedings of the 32nd International Conference on Machine Learning},
  series = {Proceedings of Machine Learning Research},
  volume = {37},
  pages = {209--217},
  year = {2015}
}

@inproceedings{jang2019,
  author = {Jang, Jennifer and Jiang, Heinrich},
  title = {{DBSCAN}++: Towards Fast and Scalable Density Clustering},
  booktitle = {Proceedings of the 36th International Conference on Machine Learning},
  series = {Proceedings of Machine Learning Research},
  volume = {97},
  pages = {3019--3029},
  year = {2019}
}

@article{fraley2007,
  author  = {Fraley, Chris and Raftery, Adrian E.},
  title   = {Bayesian Regularization for Normal Mixture Estimation and Model-Based Clustering},
  journal = {Journal of Classification},
  year    = {2007},
  volume  = {24},
  number  = {2},
  pages   = {155--181},
  doi = {10.1007/s00357-007-0004-5}
}

@book{hartigan1975,
  author    = {Hartigan, John A.},
  title     = {Clustering Algorithms},
  publisher = {Wiley},
  address   = {New York},
  year      = {1975}
}

@article{hartigan1981,
  author  = {Hartigan, J. A.},
  title   = {Consistency of Single Linkage for High-Density Clusters},
  journal = {Journal of the American Statistical Association},
  year    = {1981},
  volume  = {76},
  number  = {374},
  pages   = {388--394},
  doi     = {10.1080/01621459.1981.10477658}
}

@article{lucic2018,
  author  = {Lucic, Mario and Faulkner, Matthew and Krause, Andreas and Feldman, Dan},
  title   = {Training {G}aussian Mixture Models at Scale via Coresets},
  journal = {Journal of Machine Learning Research},
  year    = {2018},
  volume  = {18},
  number  = {160},
  pages   = {1--25}
}

@inproceedings{kim2016,
  author    = {Kim, Jisu and Chen, Yen-Chi and Balakrishnan, Sivaraman and Rinaldo, Alessandro and Wasserman, Larry},
  title     = {Statistical Inference for Cluster Trees},
  booktitle = {Advances in Neural Information Processing Systems},
  volume    = {29},
  publisher = {Curran Associates, Inc.},
  year      = {2016}
}

@article{wade2023,
  author  = {Wade, Sara},
  title   = {{B}ayesian Cluster Analysis},
  journal = {Philosophical Transactions of the Royal Society A: Mathematical, Physical and Engineering Sciences},
  year    = {2023},
  volume  = {381},
  number  = {2247},
  pages   = {20220149},
  doi     = {10.1098/rsta.2022.0149}
}

@article{cuevas2001,
  author  = {Cuevas, Antonio and Febrero, Manuel and Fraiman, Ricardo},
  title   = {Cluster Analysis: A Further Approach Based on Density Estimation},
  journal = {Computational Statistics \& Data Analysis},
  year    = {2001},
  volume  = {36},
  number  = {4},
  pages   = {441--459},
  doi     = {10.1016/S0167-9473(00)00052-9}
}

\end{document}